\documentclass[letterpaper, 10 pt, conference]{ieeeconf}  

\IEEEoverridecommandlockouts                              
\usepackage[pdftex]{color,graphicx}
\usepackage{subfigure}

\usepackage{graphics} 
\usepackage{epsfig} 
\usepackage{times} 
\usepackage{amsmath} 
\usepackage{amssymb}  

\title{\LARGE \bf
Stability and Transparency Analysis of a Bilateral Teleoperation in Presence of Data Loss 
} 

\author{ A. Bakhshi, H.A. Talebi, A.A. Suratgar, and M. Abdeetedal 
\thanks{ Authors are with the Department of Electrical Engineering, AmirKabir University of  Technology, Tehran, Iran. Email:\{alibakhshi21,alit,a-suratgar,metedal\}@aut.ac.ir.
        }%
}
\newtheorem{asm}{Assumption}[section]
\newtheorem{prop}{Property}
\newtheorem{theo}{Theorem}[section]
\newtheorem{lem}{Lemma}
\begin{document}

\maketitle
\thispagestyle{empty}
\pagestyle{empty}

\begin{abstract}

This paper presents a novel approach for stability and transparency analysis for bilateral teleoperation in the presence of data loss in communication media. A new model for data loss is proposed based on a set of periodic continuous pulses and its finite series representation. The passivity  of the overall  system is shown using wave variable approach including the newly defined model for data loss. Simulation results are presented to show the effectiveness of the proposed approach.
\end{abstract}

\section{Introduction}

In a teleoperation system, the operator can perform a task in a remote environment via slave robot that commanded by the master robot. Tele-manipulation  tasks are done by sending position, velocity, and/or force information remotely. The main applications are outer space explorations \cite{1}, handling of toxic materials \cite{2} , and minimally invasive surgery \cite{w}. In bilateral teleoperation, a communication network structured the connection between  master and the slave. Time delay and data loss occur in communication channel between the master and slave sites. Wave variable and smith predictor methods are among the common control schemes for compensation of time delay and data loss. 

Anderson and Spong used scattering transformation, network theory and the passivity to provide the stability of a teleoperation system in presence of constant time delay \cite{3}. Niemeyer and Slotine introduced wave variables \cite{4}. In this wave-based framework, the scattered wave variables are communicated via the delayed transmission lines instead of the typical power-conjugated variables such as force and velocity. The use of the wave scattering approach solved the destabilizing effects incurred in the transmission lines by passifying the communication channels independent of the amount of the delay.However, this approach is not readily applicable to time-varying delays. Moreover there is no guarantee in performance in the presence of even constant communication delay.

Since position information does not pass through the communication channels, system faces position drift and raising tracking error. furthermore the time varying delay causes distorted velocity signal and results in tracking error. In order to compensate the effects of time varying delay, wave variables are sent along their integrations \cite{5} , \cite{6} . In this method the main signal and its integral are received by the slave and the position is then calculable. moreover in \cite{6},  pre mentioned approach is extended by sending power signal of the wave variable along side other signals. In \cite{7}, instability is avoided by defining new gains in communication channels. These gains result in stability of the system, although cause poor performance. In \cite{8}, delayed position signals are also used but this strategy causes tracking error as well. Similar approach has been used in \cite{9} to converge the velocity error to zero. How ever, it failed to to provide zero position tracking error.

In \cite{10}, a proper state feedback is designed and  passive input/output system is defined such that it includes the position and velocity information. This approach result in good performance in presence of time delay. To avoid the adverse effects of wave distortion due to time varying delays and data losses, a novel solution  based on the digital reconstruction of the wave variables is proposed in \cite{11}. This approach introduces the use of  buffering and interpolation scheme that preserves the passivity of the system which reduces the tracking error under time-varying delays and packet losses.
In \cite{12}, the passivity of the system has been shown in the presence of data loss by considering zero value Wave in a discrete communication. As long as data loss causes poor performance, considering data reconstruction methods is mandatory. In \cite{13}, the effects of time delay and packet loss are compensated by estimating the received wave variables. This technique is based on Smith predictor and Kalman filter. Several researches have been done on compensating the effects of time delay in communication channels though most of which result in  poor performance in the presence of data loss.

In network control system (NCS) literature, data loss often is modelled by using three main categories. In \cite{14}, data loss has been modelled as jump linear systems with Markov chains. In \cite{15}, modelling has been done using asynchronous dynamical system(ADS). In \cite{16}, random sample system has been used for modelling.

In this paper, a robust and high performance teleoperation system in the presence of data loss and different initial conditions is presented. A new model for data loss is proposed based on a set of periodic continuous pulses and its finite series representation and a state feedback control law for master and slave manipulators is designed. The presented control laws contains  both position and velocity information. The new architecture, built within the passivity framework, provides good transparency which is measured in terms of position tracking abilities of the bilateral system. This configuration provides robust performance against network effects such as packet losses and reordering which shows the ability of the architecture to teleoperate over unreliable communication networks .

The structure of this paper is as follows. In Section II, our proposed model for data loss is presented. The data loss is presented as a set of periodic continuous pulses and its finite series representation. In Section III, the structure of the teleoperation system in presence of the data loss is proposed. In Section IV, the stability of the teleoperation system is shown by defining a Lyapunov function candidate and obtain conditions to have acceptable performance. In Section V, the validation of our control approach is investigated via simulations and conclusion is presented in Section VI.
\section{Data loss modelling }
 In this paper, we propose a new continuous time model for data loss.
\begin{asm}
Data loss occurs in periodic manner which can be written as a train of pulse as shown  in Fig. \ref{trainpulse}.
\end{asm}
\begin{figure}
\centering
\includegraphics[scale=0.30]{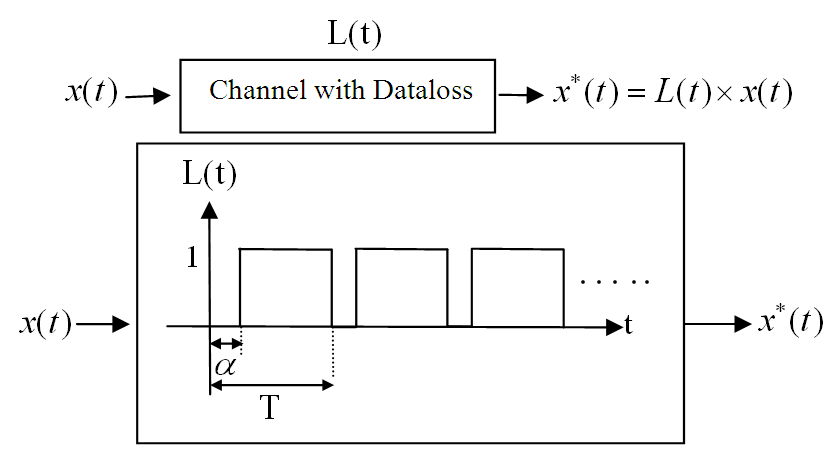}
\caption{Periodic data loss modelled as a train pulse}
\label{trainpulse}
\end{figure}

Where $(.)^*$ indicates signals which passes through communication channel in which data loss occurs, $\frac{\alpha}{T}$ is loss rate, $\alpha  \ll T$ and  $L(t)$ is a periodic function of time and.

\begin{asm}
 $L(t)$ can be written as a finite Fourier series which is continuous and periodic where :
\begin{equation}
\begin{split}
L(t) &= \sum\limits_{n = 1}^N {{a_n}\cos (n{w_0}t)}  + {b_n}\sin (n{w_0}t)\\
{a_n}& = \frac{2}{T}{\smallint _T}L(t) \times \cos (n{w_0}t)dt\\
{b_n} &= \frac{2}{T}{\smallint _T}L(t) \times \sin (n{w_0}t)dt
\end{split}
\label{1}
\end{equation}
\end{asm}
Substituting the loss rate $\frac{\alpha }{T}$ and frequency of the train pulse signal ${W_0} = \frac{{2\pi }}{T}$ into  (\ref{1}), we obtain
\begin{equation}
\begin{split}
{a_n} &= \frac{{ - 1}}{{\pi n}}\sin (\frac{{2\pi n\alpha }}{T})\\
{b_n} &= \frac{1}{{\pi n}}[\cos (\frac{{2\pi n\alpha }}{T})\, - {( - 1)^n}]\\
L(t) &= \sum\limits_{n = 1}^N {{a_n}\cos (n{w_0}t)}  + {b_n}\sin (n{w_0}).
\end{split}
\label{2}
\end{equation}
In fact, L(t) is a continues function that describes data loss as a periodic phenomenon.
\begin{figure}
\centering
\includegraphics[scale=0.30]{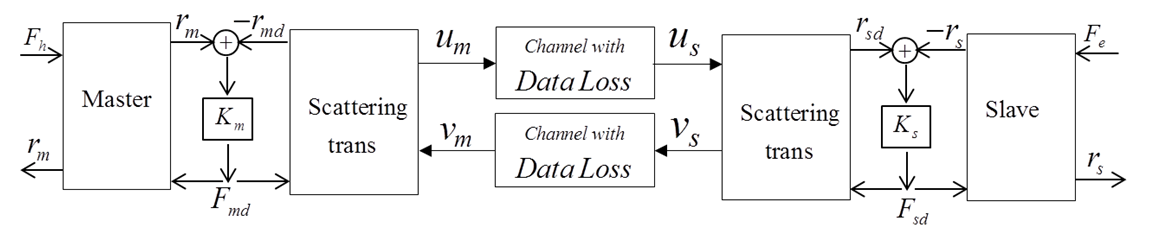}
\caption{New teleoperation system according to wave variables in presence of data loss in communication channel. }
\label{new}
\end{figure}

\section{COORDINATION ARCHITECTURE FOR BILATERAL TELEOPERATION}
The new presented model for data loss is implemented on the teleoperation structure in  \cite{10}. Master and slave dynamics are considered as follows :
\begin{equation}
\begin{array}{l}
{M_m}({q_m}){{\ddot q}_m} + {C_m}({q_m},{{\dot q}_m}){{\dot q}_m} + {g_m}({q_m}) = {J^T}{F_h} + {\tau _m}\\
{M_s}({q_s}){{\ddot q}_s} + {C_s}({q_s},{{\dot q}_s}){{\dot q}_s} + {g_s}({q_s}) = {\tau _s} - {J^T}{F_e}
\end{array}
\label{3}
\end{equation}
where ${q_m},{q_s}$ are $n \times 1$ master and slave robots joint variables vectors, ${\dot q_m},{\dot q_s}$  are $n \times 1$  joint velocity vectors, ${\tau _m},{\tau _s}$ are $n \times 1$ applied torque vectors, ${M_s},{M_m}$ are positive definite inertial $n \times n$ matrices, C is $n \times n$ Coriolis matrices and g is $n \times 1$ gravitational vectors.

It is worth mentioning that some properties of the robot structures are as fallows
\\
 
 \begin{prop}
 The inertia matrix $M$ is symmetric positive definite and there exist some positive constants ${m_1},{m_2}$ such that ${m_1}I < M < {m_2}I$ \\

 \label{p1}
 \end{prop}
 \begin{prop}
 Using  the  Christoffel  symbols, the  matrix $\dot M - 2C$ is skew-symmetric. \\
 
 \end{prop}
 The master and slave control torques are given by
 \begin{equation}
 \begin{split}
{\tau _m} &=  - {M_m}\lambda {{\dot q}_m} - {C_m}\lambda {q_m} + {g_m} + {{\bar \tau }_m}\,\,\,\,\\
{\tau _s} &=  - {M_s}\lambda {{\dot q}_s} - {C_s}\lambda {q_s} + {g_s} + {{\bar \tau }_s}
\end{split}
\label{4}
 \end{equation}
Where ${\tau _m},{\tau _s}$ are master and slave actuators torques respectively. By substituting (\ref{4}) in (\ref{3}) the  master and slave dynamic can be expressed 
\begin{equation}
\begin{split}
{M_m}{{\dot r}_m} + {C_m}{r_m}& = {F_h} + {{\bar \tau }_m} = {{\tau '}_m}\\
{M_s}{{\dot r}_s} + {C_s}{r_s} &=  - {F_e} + {{\bar \tau }_s} = {{\tau '}_s}
\end{split}
\label{5}
\end{equation} 
Where ${r_m},{r_s}$ are defined as in (\ref{6})
\begin{equation}
\begin{split}
{r_m}& = {{\dot q}_m} + \lambda {q_m}\\
{r_s} &= {{\dot q}_s} + \lambda {q_s}\,\,\,.
\end{split}
\label{6}
\end{equation}
In fact, ${r_m},{r_s}$ are new outputs of the teleoperation system, where
\begin{equation}
\smallint _0^t\,{\tau '_m}^T{r_m} = \smallint _0^t{\tau '_s}^T{r_s} = 0\,\,\,.
\label{7}
\end{equation} 
Hence according to (\ref{7}) the new teleoperation system is lossless.
where $({\tau '_m},{r_m})$ and $({\tau '_s},{r_s})$ indicate new input/output of the system. Since master and salve are passive based on  new input/output definitions  hence just  the stability of channel should be proven. As in \cite{10}  the new structure of the teleoperation systems with respects to new definitions is illustrated in Fig. \ref{new}

Where variables in Fig. \ref{new} are as follows:
\begin{equation}
\begin{split}
{u_m} = \frac{1}{{\sqrt {2b} }}({F_{md}} + b{r_{md}})\,\,{V_m}& = \frac{1}{{\sqrt {2b} }}({F_{md}} - b{r_{md}})\\
{u_s} = \frac{1}{{\sqrt {2b} }}({F_{sd}} + b{r_{sd}})\,\,{V_s} &= \frac{1}{{\sqrt {2b} }}({F_{sd}} - b{r_{sd}})
\end{split}
\label{8}
\end{equation}
where ${r_{md}},{r_{sd}}$ are the signals obtained from scattering transformation. Control signals ${F_{sd}}({\bar \tau _s})$ and ${F_{md}}( - {\bar \tau _m})$ are obtained in (\ref{9}).
\begin{equation}
\begin{split}
{F_{sd}}& = {K_s}({r_{sd}} - {r_s})\,\,\,\\
{F_{md}} &= {K_m}({r_{md}} - {r_m})
\end{split}
\label{9}
\end{equation}
Where $b,{K_s}$ and ${K_m}$ are positive definite diagonal matrices which should be chosen properly. According to Fig. \ref{trainpulse} Input signals of communication channel ${u_m},{v_s}$ and output signals ${u_s},{v_m}$ in presence of data loss are as in Fig. \ref{waverel} and in \ref{10} :
\begin{figure}
\centering
\includegraphics[scale=0.30]{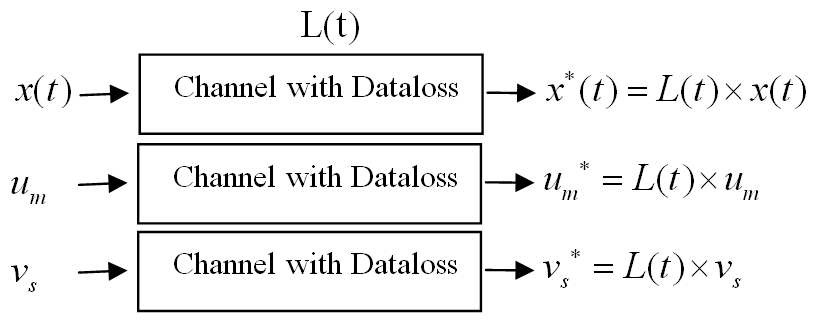}
\caption{Relations between wave variables in presence of data loss. }
\label{waverel}
\end{figure}
\begin{equation}
\begin{split}
{u_m}^ * & = {u_m}L = {u_m}\sum\limits_{n = 1}^N {{a_n}\cos (n{w_0}t)}  + {b_n}\sin (n{w_0}) = {u_s}\\
{v_s}^ *  &= {v_s}L = {v_s}\sum\limits_{n = 1}^N {{a_n}\cos (n{w_0}t)}  + {b_n}\sin (n{w_0}) = {v_m}\,\,.
\end{split}
\label{10}
\end{equation}  
In order to determine ${K_m}\,$ and ${K_s}$, using presented data loss modelling and substituting (\ref{10}) into (\ref{8}) result in (\ref{11})  .
 \begin{equation}
 F_{md}^* + br_{md}^* = {F_{sd}} + b{r_{sd}}
 \label{11}
 \end{equation}
 and by substituting (\ref{9}) in (\ref{11}) ,(\ref{12}) result in.
 \begin{equation}
 (b + {K_s}){r_{sd}} = (b - {K_m})r_{md}^* + {K_m}r_m^* + {K_s}{r_s}
 \label{12}
\end{equation}  
and similarly (\ref{13}) is obtained.
\begin{equation}
(b+{K_m}){r_{md}} = (b- {K_s})r_{sd}^* + {K_s}r_s^* + {K_m}{r_m}
\label{13}
\end{equation}
In (\ref{12}) and (\ref{13}) according to the effects of $r_{md}^*$ and $r_{sd}^*$ on $r_{sd}$ and $r_{sd}$ respectively  and the wave reflecting phenomenon, ${K_m}\,$and$,{K_s}$ are selected equal to b, which simplifies (\ref{12}),(\ref{13}) to :
\begin{equation}
\begin{split}
{K_m}& = {K_s} = b\\
2{r_{sd}}& = r_m^* + {r_s}\\
2{r_{md}} &= r_s^* + {r_m}\,\,\,.
\end{split}
\label{14}
\end{equation}

\section{STABILITY ANALYSIS}
 
In order to analyse stability of the system, first  tracking errors are defined as follows :
\begin{equation}
\begin{split}
{e_m} &= q_m^* - q_s^{}\\
{e_s} &= q_s^* - q_m^{}
\end{split}
\label{15}
\end{equation}
The main goal of the controller is that defined errors are converging to zero which result in highly improve transparency of the system. 
\begin{asm}

The environment and operator are assumed to be passive with inputs ${r_m}$,${r_s}$. 
\label{passive}
\end{asm}
\begin{asm}
The environment's and operator's forces are assumed to be bounded.
\end{asm}
\begin{asm}
${r_{md}}$,${r_{sd}}$ are assumed to be zero for $t < 0$
\\
\end{asm}
\begin{theo}
A teleoperation system described as (\ref{5}), (\ref{8}) and (\ref{14}) in presence of data loss described as (\ref{2}) is stable and errors defined as (\ref{15}) are bounded. 
\label{t1}
\end{theo}
\begin{proof}
Let us define a Lyapunov function candidate as :
\begin{equation}
\begin{split}
V =& \frac{1}{2}(r_m^TM_m^{}{r_m} + r_s^TM_s^{}{r_s} + e_m^TK_1^{}{e_m} + e_s^TK_1^{}{e_s}) + \\
&\smallint _0^t(F_e^T{r_s} - F_h^T{r_m})dt + \smallint _0^t(F_{md}^T{r_{md}} - F_{sd}^T{r_{sd}})dt
\end{split}
\label{16}
\end{equation}
where V(0)=0 and the first four terms are positive according to positive definiteness of  ${M_s},{M_m},{K_1},{K_2}$ and property \ref{passive}. Hence  :
$\smallint _0^t(F_e^T{r_s})dt \ge 0\,\,$and$ - \smallint _0^t(F_h^T{r_m})dt \ge 0$
Hence just  should be shown that :
\begin{equation*}
 {V_1} = \smallint _0^t(F_{md}^T{r_{md}} - F_{sd}^T{r_{sd}})dt\,\, > 0
\label{17}
\end{equation*}
From (\ref{8}) can be shown : 
\begin{equation}
{F_{md}} = \sqrt {\frac{b}{2}} ({u_m} + {v_m})\,\,\,\,,\,\,\,{F_{sd}} = \sqrt {\frac{b}{2}} ({u_s} + {v_s})
\label{18}
\end{equation}
\begin{equation}
\begin{array}{l}
{r_{md}} = \frac{1}{{\sqrt {2b} }}({u_m} - {v_m})\,\,,\,\,\,\,
{r_{sd}} = \frac{1}{{\sqrt {2b} }}({u_s} - {v_s})\,
\end{array}
\label{19}
\end{equation}
Using (\ref{18}) and (\ref{19}) :
\begin{equation}
\begin{array}{l}
\begin{array}{l}
V_1\,\,\, =\int\limits_0^t {\frac{1}{2}(u_m^2 - v_m^2)}  - \int\limits_0^t {\frac{1}{2}(u_s^2 - v_s^2)dt} 
\end{array}
\end{array}
\label{20}
\end{equation}
And substituting (\ref{10}) and (\ref{20}) :
\begin{equation}
\begin{split}
\,\,{V_1} &=\int\limits_0^t {\frac{1}{2}(u_m^2 - v{{_s^*}^2})dt}  - \frac{1}{2}\int\limits_0^t {(u{{_m^*}^2} - v_s^2)} dt\\
& = \int\limits_0^t {\frac{1}{2}(u_m^2 - u{{_m^*}^2})dt}  + \frac{1}{2}\int\limits_0^t {(v_s^2 - v{{_s^*}^2})} dt
\end{split}
\end{equation}
 According to data loss model in Fig. \ref{trainpulse} it is clear that $\,u_m^2 - {(u_m^*)^2} > 0\,$ , $v_s^2 - {(v_s^*)^2} > 0$ and therefore the candidate Lyapunov function is positive definite.The time derivative of $V$ is :\\
  \begin{equation*}
\begin{split}
 \dot V =& r_m^TM_m^{}{{\dot r}_m} + \frac{1}{2}r_m^T\dot M_m^{}{r_m} + r_s^TM_s^{}{{\dot r}_s}\\
&+ \frac{1}{2}r_s^T\dot M_s^{}{r_s} + \dot e_m^TK_1^{}{e_m} + \dot e_s^TK_2^{}{e_s}\\
 &+ F_e^T{r_s} - F_h^T{r_m} + F_{md}^T{r_{md}} - F_{sd}^T{r_{sd}}\,\\
 =& r_m^T( - {C_m}{r_m} + {F_h} - {F_{md}})\, + \frac{1}{2}r_m^T\dot M_m^{}{r_m}\\
&+ r_s^T( - {C_s}{r_s} - {F_e} + {F_{sd}}) + \frac{1}{2}r_s^T\dot M_s^{}{r_s}\, +\, \,\dot e_m^TK_1^{}{e_m}\\
&+ \dot e_s^TK_2^{}{e_s} + F_e^T{r_s} - F_h^T{r_m} + F_{md}^T{r_{md}} - F_{sd}^T{r_{sd}}\,\\
=& r_m^T({F_h} - {F_{md}})\, + \,r_s^T( - {F_e} + {F_{sd}}) + \,\,\dot e_m^TK_1^{}{e_m}\\
&+ \dot e_s^TK_2^{}{e_s} + F_e^T{r_s} - F_h^T{r_m} + F_{md}^T{r_{md}} - F_{sd}^T{r_{sd}}\,\\
 =&  - {({r_{md}} - {r_m})^T}b({r_{md}} - {r_m}) - {({r_s} - {r_{sd}})^T}b({r_s} - {r_{sd}})\\
&+ \,\dot e_m^TK_1^{}{e_m} + \dot e_s^TK_2^{}{e_s}\,\,\,\,\,\,\,\,\,\,\,\,\,\,\,\,\,\,\,\,\,\,\,\,\,\,\,\,\,\,
\end{split}
 \label{21}
 \end{equation*}
 Using (\ref{6}),(\ref{14}) and (\ref{15}) :
 \begin{equation}
 \begin{split}
r_s^* - r_m^{} &= \dot q_s^* + \lambda q_s^* - {{\dot q}_m} - \lambda {q_m} = {{\dot e}_s} + \lambda {e_s}\\
r_s^{} - r_m^* &= \dot q_s^{} + \lambda q_s^{} - {{\dot q}^*}_m - \lambda {q^*}_m = {{\dot e}_m} + \lambda {e_m}
\end{split}
 \label{22}
 \end{equation}
 \begin{equation*}
\begin{split}
\dot V =&  - \frac{1}{4}(r_s^*\, - {r_m})b(r_s^* - {r_m}) - \frac{1}{4}(r_s^{}\, - {r^*}_m)b(r_s^{} - {r^*}_m)\,\\
+& \dot e_m^TK_1^{}{e_m} + \,\,\dot e_s^TK_2^{}{e_s}\,\\
 =&  - \frac{1}{4}{({{\dot e}_s} + \lambda {e_s})^T}b({{\dot e}_s} + \lambda {e_s})\\
 -& \frac{1}{4}{({{\dot e}_m} + \lambda {e_m})^T}b({{\dot e}_m} + \lambda {e_m}) + \dot e_m^TK_1^{}{e_m} + \,\,\dot e_s^TK_2^{}{e_s}\,\\
 = & - \frac{1}{4}(\dot e_s^Tb{{\dot e}_s} + 2\dot e_s^Tb\lambda {e_s} + e_s^T\lambda b\lambda {e_s} + \dot e_m^Tb{{\dot e}_m}\\
 +& 2\dot e_m^Tb\lambda {e_m} + e_m^T\lambda b\lambda {e_m}) + \dot e_m^T{K_1}{e_m} + \dot e_s^T{K_2}{e_s}
\end{split}
\label{23}
 \end{equation*}
 Since $b$ and $\lambda$ are positive definite diagonal matrices.
 
 Hence ${(\dot e_i^Tb\lambda {e_i})^T} = \dot e_i^Tb\lambda {e_i}\,i = m,s$
and  \\
\begin{equation*}
 \begin{split}
\dot V &=  - \frac{1}{4}(\dot e_s^Tb{{\dot e}_s} + 2\dot e_s^Tb\lambda {e_s} + e_s^T\lambda b\lambda {e_s}) + \dot e_s^T{K_2}{e_s}\\
& - \frac{1}{4}(\dot e_m^Tb{{\dot e}_m} + 2\dot e_m^Tb\lambda {e_m} + e_m^T\lambda b\lambda {e_m}) + \dot e_m^T{K_1}{e_m}
\end{split}
\end{equation*}

 by choosing  ${K_1} = {K_2} = \frac{{\lambda b}}{2}$ , extra terms are omitted and $\dot V $ is simplified  to
\begin{equation}
\dot V =  - \frac{1}{4}({\dot e_s}^Tb{\dot e_s} + e_s^T\lambda b\lambda {e_s} + {\dot e_m}^Tb{\dot e_m} + e_m^T\lambda b\lambda {e_m}) \le 0
\label{24}
\end{equation}
As V is positive definite  and  $\dot V $ is negative semi definite so the system is stable and the errors are bounded.
Using property\ref{p1}, ${r_m}$ and ${r_s}$ remain bounded.
\begin{equation}
\begin{split}
{r_m}& = {{\dot q}_m} + \lambda {q_m} \Rightarrow s{q_m} + \lambda {q_m} = {r_m} \Rightarrow \,{q_m} = \frac{{{r_m}}}{{s + \lambda }}\\
{r_s} &= {{\dot q}_s} + \lambda {q_s} \Rightarrow \,s{q_s} + \lambda {q_s} = {r_s} \Rightarrow {q_s} = \frac{{{r_s}}}{{s + \lambda }}
\end{split}
\label{26}
\end{equation}
According to (\ref{26}) and boundedness of ${r_m}$ and ${r_s}$, boundedness of ${q_m},{q_s}$,${\dot q_m},{\dot q_s}$ are  concluded. Hence 
\begin{equation}
{q_m}, {q_s}, {r_m}, {r_s}, {\dot q_m},{\dot q_s}\,\mbox{  is bounded.} \,\, 
\label{27}
\end{equation}
\begin{theo}
Teleoperation system described with (\ref{3}),(\ref{4}),(\ref{5}),(\ref{6}),(\ref{8}),(\ref{9}) and (\ref{14}) which is illustrating in Fig.\ref{new},  ${\ddot q_m},{\ddot q_s}$ are bounded if environment and operator are passive.

\label{t2}
\end{theo}
\begin{proof}
From boundedness of ${r_m},{r_s}\,\,\,is\,\,\,bounded$ we can conclude the boundedness of ${r_{sd}},{r_{md}}$ based on (\ref{14}).
Using (\ref{9}) ${F_{md}}$ and ${F_{sd}}$ remain bounded. Hence ${M_m},{M_s},{\dot q_m}$ and ${\dot q_s}$ are bounded.
From (\ref{3}) 
\begin{equation}
{\ddot q_m}\mbox{ and } {\ddot q_s}\mbox{ are bounded}
\label{28}
\end{equation}
\end{proof}
Since the system is non-autonomos, the Barballat's lemma \cite{17} is used . 
\begin{lem}
If V satisfies following three conditions:
\begin{itemize}
\item $V$ is lower bounded.
\item $\dot V \le 0$
\item $\dot V$ is uniformly continuous
\end{itemize}
Then \[\dot V \to 0\,\,as\,\,t \to \infty \]
\end{lem}
At this stage uniformly continuity of $\dot V$  are proven. Using \cite{17} $\dot V$ is uniformly continuous if $\ddot V$ is always bounded. Hence it should be shown that ${\ddot e_s},{\ddot e_m},{\dot e_m},{\dot e_s}$ in what conditions remain bounded. 

Since ${e_m},{e_s}$ are bounded , Using (\ref{2}) and Fig. \ref{trainpulse} then lost signals are
$q_m^* = L(t) {q_m}\,,\,\,\,q_s^* = L(t) {q_s}$\\
Using the boundedness of the first and third term of (\ref{31}) it should be proven  that $\dot L(t) \times {q_m}$ is bounded.
 \begin{equation}
\begin{split}
{e_m} &= q_m^* - {q_s} = L(t) {q_m} - {q_s}\,\\
 \to \,{{\dot e}_m}& = \,{{\dot q}_m}L(t) + \,{q_m} \dot L(t)\, + \,{{\dot q}_s}
\end{split}
 \label{31}
 \end{equation}
   As long as ${q_m}$ is bounded, the boundedness of $\dot L(t)$ should be proven. Hence the train pulse are written  by Fourier series with finite terms and derivative of that  it with respect to time is as follows :
   \begin{equation*}
  \begin{split}
L &= \sum\limits_{n = 1}^N {\frac{{ - 1}}{{\pi n}}\sin (\frac{{2\pi n\alpha }}{T})\cos (n{w_0}t)} \\
\,\,\,\,\,\,\,\,\,\,\, &+ \frac{1}{{\pi n}}[\cos (\frac{{2\pi n\alpha }}{T})\, - {( - 1)^n}]\sin (n{w_0}t)\\
\,\,\,\,\dot L\,& = \sum\limits_{n = 1}^N {\frac{2}{T}\sin (\frac{{2\pi n\alpha }}{T})\sin (\frac{{2n\pi }}{T}t)} \\
\,\,\,\,\,\,\,\,\, &+ \frac{2}{T}[\cos (\frac{{2\pi n\alpha }}{T})\, - {( - 1)^n}]\cos (\frac{{2n\pi }}{T}t)
\end{split}
      \end{equation*}
The boundedness of $\dot L$ should be investigated at the point $t = \alpha $.\\
\begin{equation*}
\begin{split}
{{\dot L}_{t = \alpha }}\, &= \sum\limits_{n = 1}^N {\frac{2}{T}\sin (\frac{{2\pi n\alpha }}{T})\sin (\frac{{2n\pi }}{T}\alpha )} \\
\,\,\,\,\,\,\,\,\,\,\, &+ \frac{2}{T}[\cos (\frac{{2\pi n\alpha }}{T})\, - {( - 1)^n}]\cos (\frac{{2n\pi }}{T}\alpha )\\
\,\,\,\,\,\,\,\,\,\,\,\, &= \,\,\frac{2}{T}\,\sum\limits_{n = 1}^N 1 \, - {( - 1)^n}\cos (\frac{{2n\pi }}{T}\alpha )
\end{split}
\end{equation*}

Since Fourier series is considered to be written with finite terms
\begin{equation}
{\dot e_m}, {\dot e_s}\,\,is\,\,bounded
\label{32}
\end{equation}

In order to investigate the boundedness of ${\ddot e_m},{\ddot e_s}$  thus \\
\begin{equation*}
\begin{split}
{e_m} &= q_m^* - {q_s} = L(t){q_m} - {q_s}\\
{{\dot e}_m} &= \dot q_m^* - {{\dot q}_s} = \dot L(t){q_m} + L(t){{\dot q}_m} - {{\dot q}_s}\\
{{\ddot e}_m} &= \ddot q_m^* - {{\ddot q}_s} = \ddot L(t){q_m} + 2\dot L(t){{\dot q}_m} + L(t){{\ddot q}_m} - {{\ddot q}_s}
\end{split}
\end{equation*}
With using the boundedness of ${q_m},L(t),\dot L(t),{\ddot q_s},{\ddot q_m}$ and ${{\dot q}_m}$ the boundedness of $\ddot L(t)$ should be shown hence :
\begin{equation*}
\begin{split}
L =& \sum\limits_{n = 1}^N {\frac{{ - 1}}{{\pi n}}\sin (\frac{{2\pi n\alpha }}{T})\cos (\frac{{2n\pi }}{T}t)} \\
\,\,\,\,\,\,\,\,\,\,\,\,\, +& \frac{1}{{\pi n}}[\cos (\frac{{2\pi n\alpha }}{T})\, - {( - 1)^n}]\sin (\frac{{2n\pi }}{T}t)\\
\,\,\,\,\,\,\dot L =& \sum\limits_{n = 1}^N {\frac{2}{T}\sin (\frac{{2\pi n\alpha }}{T})\sin (\frac{{2n\pi }}{T}t)} \\
\,\,\,\,\,\,\,\,\,\,\,\,\, +& \frac{2}{T}[\cos (\frac{{2\pi n\alpha }}{T})\, - {( - 1)^n}]\cos (\frac{{2n\pi }}{T}t)\\
\,\,\,\,\,\,\ddot L=& \sum\limits_{n = 1}^N {\frac{2}{T}\frac{{2n\pi }}{T}\sin (\frac{{2\pi n\alpha }}{T})\cos (\frac{{2n\pi }}{T}t)} \\
\,\,\,\,\,\,\,\,\,\,\,\, -& \frac{2}{T}\frac{{2n\pi }}{T}[\cos (\frac{{2\pi n\alpha }}{T})\, - {( - 1)^n}]\sin (\frac{{2n\pi }}{T}t)
\end{split}
\end{equation*}
According to variation of $L(t)$ at $t = \alpha $ it is sufficient to investigate the boundedness neighbourhood of this point.
\begin{equation}
\begin{split}
\ddot L{\,_{t = \alpha }}\, =& \sum\limits_{n = 1}^N {\frac{2}{T}\frac{{2n\pi }}{T}\sin (\frac{{2\pi n\alpha }}{T})\cos (\frac{{2n\pi }}{T}\alpha )} \\
 \,\,\,\,\,\,\,\,\,\,\,\,\,\,\,\,\,\,\,\, -& \frac{2}{T}\frac{{2n\pi }}{T}[\cos (\frac{{2\pi n\alpha }}{T})\, - {( - 1)^n}]\sin (\frac{{2n\pi }}{T}\alpha )\\
    \,\,\,\,\,\,\,\,\, \, \,\,\,\,\,\,=& \sum\limits_{n = 1}^N { - \frac{2}{T}\frac{{2n\pi }}{T}[ - {{( - 1)}^n}]\sin (\frac{{2n\pi }}{T}\alpha )} \\
    \,\,\,\,\,\,\,\,\, \,\,\,\,\,\, =& \frac{{4\pi }}{{{T^2}}}\sum\limits_{n = 1}^N {n[{{( - 1)}^n}]\sin (\frac{{2n\pi }}{T}\alpha )} \,\,\, \to \,\,\,\,\ddot e\,\,\mbox{is bounded .}
\end{split}
\label{33}
\end{equation} 
results in converging of this series to zero as $\alpha$ converges to zero. Hence, $\ddot L$ remains bounded.
Using (\ref{32}) and (\ref{33}), $\ddot V$ remains bounded. Hence $\dot V$ is uniformly continuous, consequently  \[\dot V \to 0\,\,as\,\,t \to \infty \] 
\end{proof}
\section{SIMULATION RESULTS}

In this section we simulate our proposed teleoperation system on a single degree of freedom system with following dynamics: 
\begin{equation*}
\begin{split}
{M_m}{{\ddot q}_m} &= {F_h} + {\tau _m}\\
{M_s}{{\ddot q}_s} &= {\tau _s} - {F_e}\,
\end{split}
\label{34}
\end{equation*}
With using (\ref{3}) the dynamics can be simplified as follow
\begin{equation*}
\begin{split}
{M_m}{{\dot r}_m} &= {F_h} - {F_{md}}\\
{M_s}{{\dot r}_s} &= {F_{sd}} - {F_e}
\end{split}
\label{35}
\end{equation*}
The environment is set to be mass, spring and damper. The operator moves the master by inserting force$ {F_h}$.
\begin{equation*}
\begin{array}{l}
{M_m} = {M_s} = 1\,,\,\,\,b = 1.2\,\,
\end{array}
\label{36}
\end{equation*}
In the first step we investigate the stability of the system. In order to investigate the dissipation of the channel, we apply a pulse signal to the master in a limited time as depicted in Fig. \ref{sim1}. All signals converge to zero hence the channel is passive and stable See Fig. \ref{sim2}.
In second step in order to investigate the performance of the system in the presence of different initial conditions we consider channel with no data loss ($\frac{\alpha }{T} = 0$) and there is just initial condition between master and slave, see Fig. \ref{sim3}. the performance remains acceptable in presence of  different initial conditions between master and slave position. Finally to investigate the performance of the system in the presence of data loss in communication media. Wave variable signals were lost in channel with period T=10(s) see Fig\ref{sim6} and  Fig\ref{sim7}. 
We assume loss rate $\frac{\alpha }{T} = 0.02,0.05$  and different initial conditions between master and slave see Fig.\ref{sim4} and  Fig.\ref{sim5}. The resulting performance is acceptable. Obviously, increasing loss rate results in poor transparency although the proposed teleoperation system remains stable with tolerable performance (even in in the presence of data loss and different initial conditions).

\begin{figure}
\centering
\includegraphics[scale=0.30]{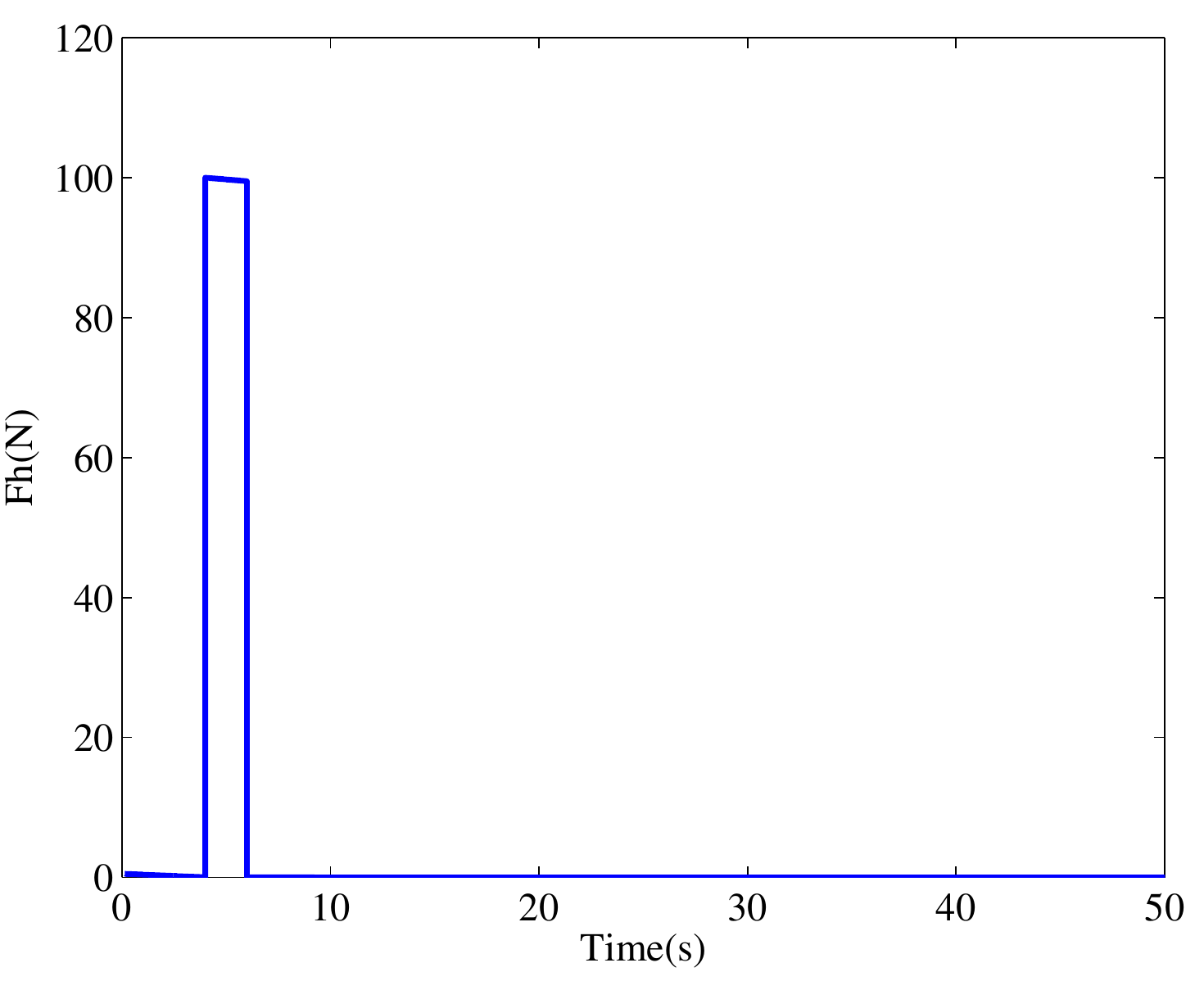}
\caption{apply a pulse to the master in a limited time }
\label{sim1}
\end{figure}

\begin{figure}
\centering
\includegraphics[scale=0.30]{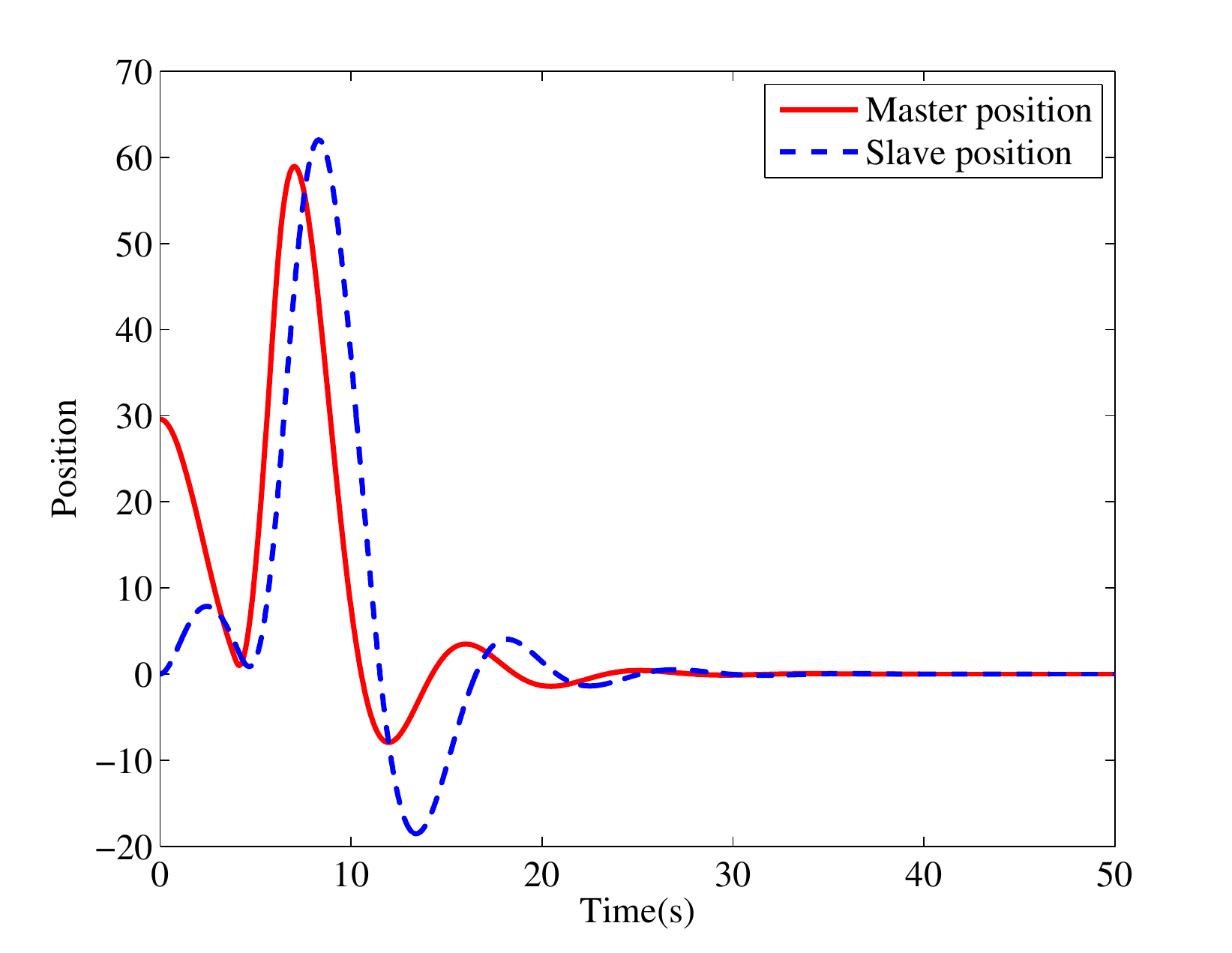}
\caption{position Tracking and damping all signal after applying a pulse in a limited time }
\label{sim2}
\end{figure}

\begin{figure}
\centering
\includegraphics[scale=0.30]{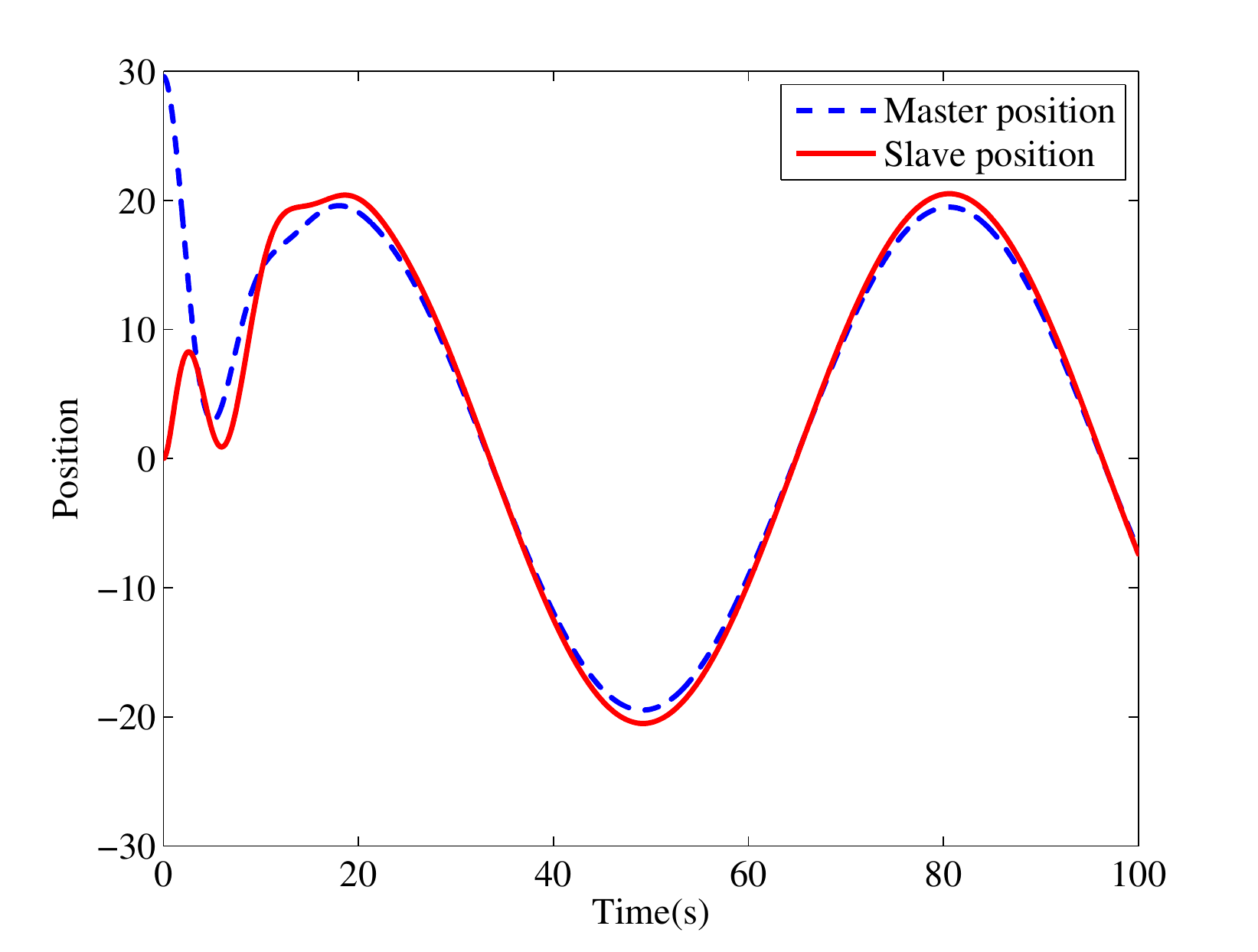}
\caption{Position tracking with different  initial conditions and $\frac{\alpha}{T}=0$}
\label{sim3}
\end{figure}

\begin{figure}
\centering
\includegraphics[scale=0.31]{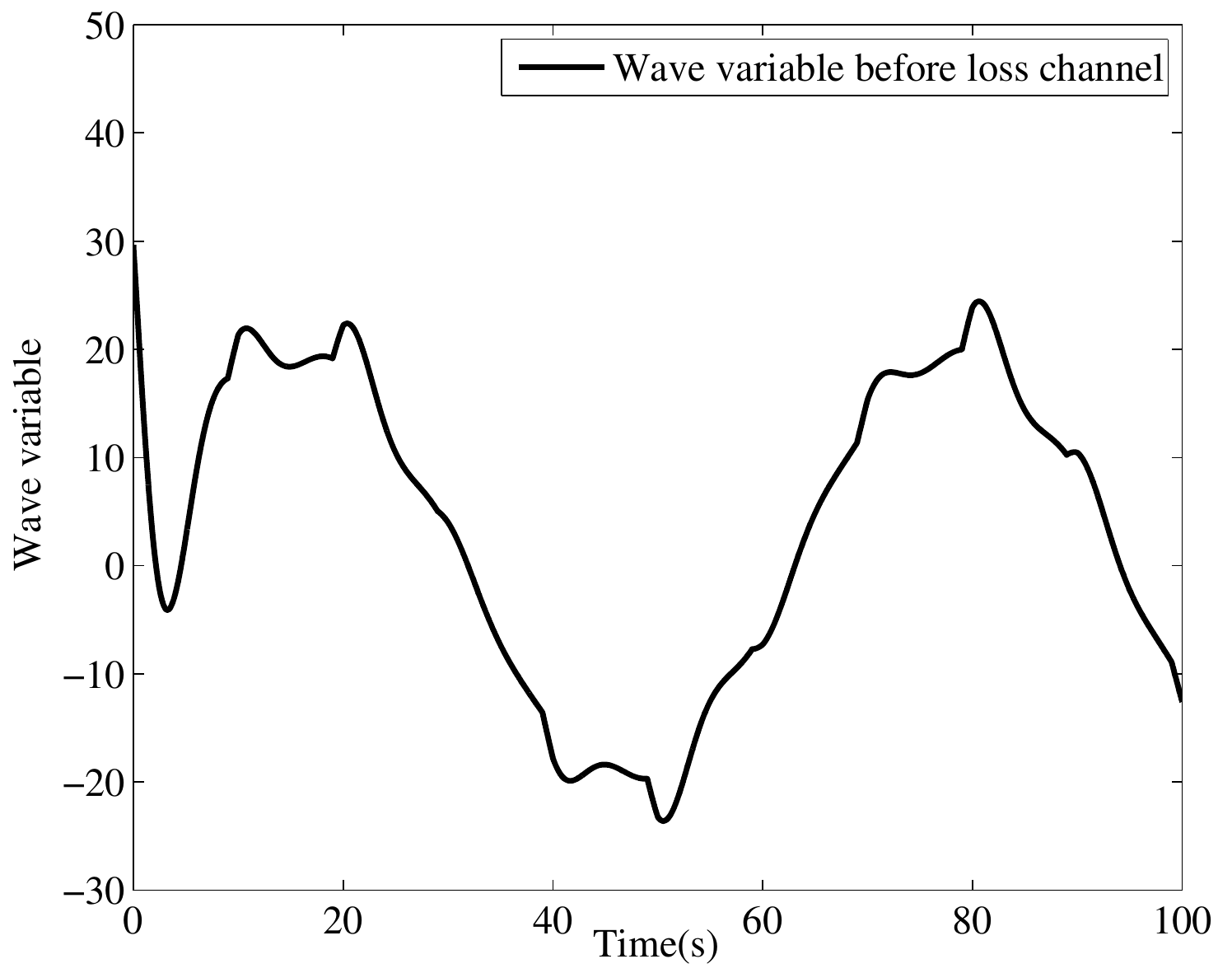}
\caption{wave variable signal before data loss channel}
\label{sim6}
\end{figure}
\begin{figure}
\centering
\includegraphics[scale=0.31]{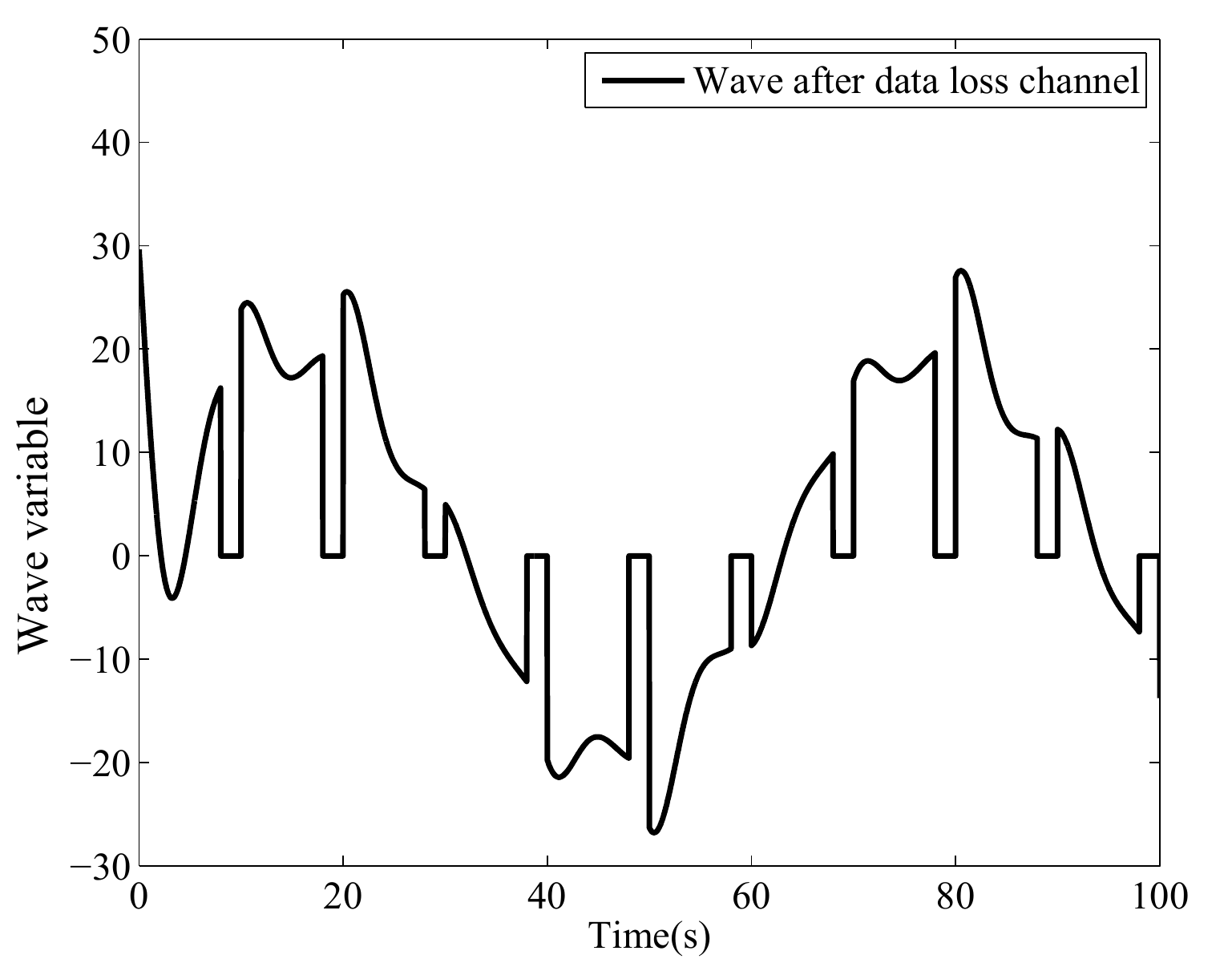}
\caption{wave variable signal after data loss channel}
\label{sim7}
\end{figure}
\begin{figure}
\centering
\includegraphics[scale=0.30]{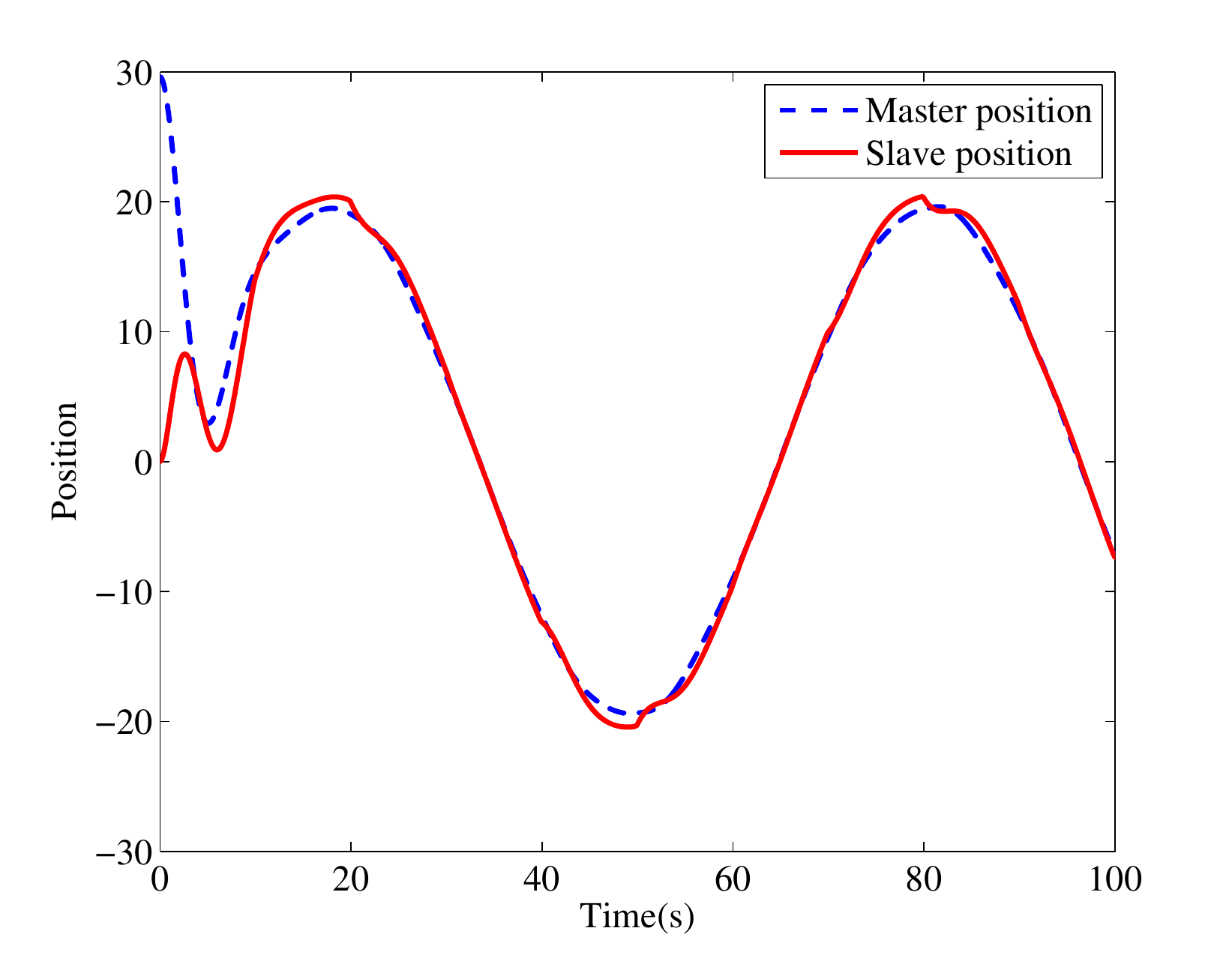}
\caption{Position tracking with different  initial conditions and $\frac{\alpha}{T}=0.02$}
\label{sim4}
\end{figure}
\begin{figure}
\centering
\includegraphics[scale=0.30]{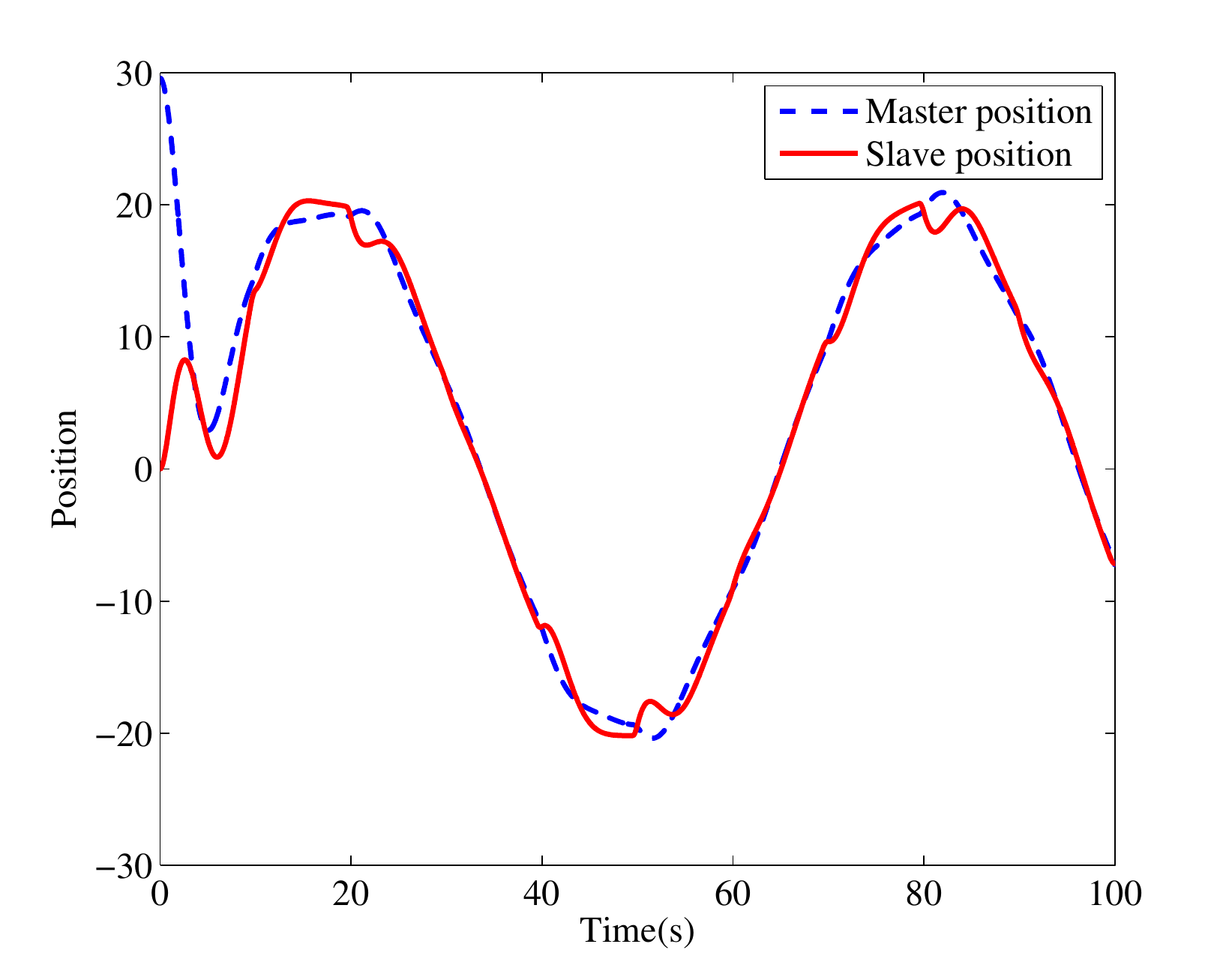}
\caption{Position tracking with different  initial conditions and $\frac{\alpha}{T}=0.05$}
\label{sim5}
\end{figure}

\section{CONCLUSIONS}
In this paper we investigated the passivity based traditional architecture  to cover position tracking in the presence of  data loss and  offset of initial conditions and  proposed a new data loss model as a set of  periodic continues pulses  and  use  a coordination architecture which uses state feedback to define a passive output for a  teleoperation system. This feedback is containing both position and velocity information and simulation results verified the usefulness  of mentioned  architecture. 
Next approach would be investigating  this architecture for a general model of data loss as a stochastic phenomenon and also improving force tracking in the presence of data loss.

\end{document}